\documentclass[runningheads]{llncs}

\usepackage{graphicx}

\usepackage{algorithm,algpseudocode}
\usepackage{algorithm}
\usepackage{algpseudocode}
\usepackage[linesnumbered,ruled,vlined,algo2e]{algorithm2e}

\usepackage{booktabs}
\usepackage{tabularx}
\usepackage{gensymb}
\usepackage{graphicx}
\usepackage{enumerate}
\usepackage{hyperref}
\usepackage{multirow}
\usepackage{xcolor}
\usepackage[labelformat=simple]{subcaption}

\usepackage{mathrsfs}
\usepackage{comment}

\usepackage{multicol}
\usepackage{graphicx,tipa}
\usepackage{amsfonts}
\usepackage{amsmath}
\usepackage{tcolorbox}
\newtheorem{observation}[theorem]{Observation}
\usepackage{algorithm,algpseudocode}
\usepackage[linesnumbered,ruled,vlined,algo2e]{algorithm2e}

\usepackage{comment}

\begin{document}
\title{Black Hole Search in Dynamic Cactus Graph}
%
\author{Adri Bhattacharya\inst{1}\thanks{Supported by CSIR, Govt. of India, Grant Number: 09/731(0178)/2020-EMR-I} \and Giuseppe F. Italiano\inst{2} \and
Partha Sarathi Mandal\inst{1}\thanks{This work was done while Partha Sarathi Mandal was in the position of Visiting Professor at Luiss University, Rome, Italy} }
\authorrunning{Bhattacharya et al.}

\institute{Indian Institute of Technology Guwahati, India\\ \email{\{a.bhattacharya, psm\}@iitg.ac.in} \and
 Luiss University, Rome, Italy\\
\email{gitaliano@luiss.it}}
\maketitle              
\begin{abstract}

We study the problem of black hole search by a set of mobile agents, where the underlying graph is a dynamic cactus. 
A black hole is a dangerous vertex in the graph that eliminates any visiting agent without leaving any trace behind. Key parameters that dictate the complexity of finding the black hole include: the number of agents required (termed as \textit{size}), the number of moves performed by the agents in order to determine the black hole location (termed as \textit{move}) and the \textit{time} (or round) taken to terminate.
This problem has already been studied where the underlying graph is a dynamic ring \cite{di2021black}. In this paper, we extend the same problem to a dynamic cactus. We introduce two categories of dynamicity, but still the underlying graph needs to be connected: first, we examine the scenario where, at most, one dynamic edge can disappear or reappear at any round. Secondly, we consider the problem for at most $k$ dynamic edges. In both scenarios, we establish lower and upper bounds for the necessary number of agents, moves and rounds.

\end{abstract}
\section{Introduction}
We study the black hole search problem (also termed as BHS) in a dynamic cactus graph, where edges can appear and disappear i.e., goes missing over time so the underlying graph remains connected. More precisely, the network is a synchronous cactus graph where one of the vertices (or nodes) is a malicious node that eliminates visiting agents without any trace of their existence upon arrival on such nodes; that node is termed as \textit{Black Hole} \cite{dobrev2006searching}. This scenario frequently arises within networked systems, particularly in situations requiring the safeguarding of agents from potential host attacks. Presently, apart from the research paper concerning ring networks \cite{di2021black}, there exists limited knowledge regarding this phenomenon when the network exhibits dynamic characteristics. Therefore, the focus of our study is to expand upon our findings in this context.

In our investigation, we consider a collection of mobile agents, all of whom execute the same algorithm synchronously. Initially, these agents are positioned at a node that is confirmed to be free from any black hole threat; these nodes are referred to as `safe nodes'. The primary objective is to efficiently determine the location of the black hole within the network in the shortest possible time while ensuring at least one agent survives and possesses knowledge of the black hole's whereabouts.\\

\noindent{\bf Related Work.} The black hole search problem is well-studied in varying underlying topologies such as rings, grids, torus, etc. This problem is first introduced by Dobrev et al. \cite{dobrev2006searching}, where they solved in static arbitrary topology. Further, they give tight bounds on the number of agents and provide cost complexity of the size-optimal solution. After this seminal paper, there has been a plethora of work done in this domain under different graph classes such as trees \cite{czyzowicz2007searching}, rings \cite{balamohan2011time,dobrev2007scattered,dobrev2008using}, tori \cite{chalopin2011black,markou2012black} and in graphs of arbitrary topology \cite{czyzowicz2006complexity,dobrev2006searching}. Moreover, two variations of this problem are mainly studied in the literature. First, when the agents are initially co-located \cite{czyzowicz2007searching} and second, when the agents are initially scattered \cite{chalopin2011black,dobrev2007scattered} in the underlying network. Now, in all the above literature, the study has been performed when the underlying graph is static.  

While most of the study has been done on static networks, very little literature is known about black hole search especially in dynamic graphs. The study of mobile agents in dynamic graphs is a fairly new area of research. Previously, the problem of exploration has been studied in dynamic rings \cite{di2020distributed,gotoh2020dynamic,mandal2020live}, torus \cite{gotoh2021exploration}, cactuses \cite{ilcinkas2021exploration} and in general graphs \cite{gotoh2021exploratio}. In addition to exploration, there are other problems related to mobile agents studied in dynamic networks such as, gathering \cite{di2020gathering}, compacting of oblivious agents \cite{das2021compacting}, dispersion of mobile agents \cite{agarwalla2018deterministic,kshemkalyani2020efficient}. Further, Flocchini et al. \cite{flocchini2012searching} studied the black hole search problem in a different class of dynamic graphs, defined as periodically varying graphs, they showed the minimum number of agents required to solve this problem is $\gamma+1$, where $\gamma$ is the minimum number of carrier stops at black holes. Di Luna et al. \cite{di2021black}, studied the black hole search problem in a dynamic ring, where they established lower bounds and give size-optimal algorithms in terms of agents, moves and rounds in two communication models. In this paper, we aim to solve a similar problem, where we want to determine the position of a black hole with the least number of agents, but in our case, we have considered the underlying topology to be a dynamic cactus graph.\\

\noindent{\bf Our Contributions.} We obtain the following results when the cactus graph has at most one dynamic edge at any round.\\
\begin{itemize}
    \item Establish the impossibility to find a black hole in a dynamic cactus with 2 agents.
    \item With 3 agents we establish lower bound of $\Omega(n^{1.5})$ rounds,  $\Omega(n^{1.5})$ moves, and we also establish upper bound of $O(n^{2})$ rounds and $O(n^{2})$ moves.
    \item With $4$ agent improved lower bounds are $\Omega(n)$ rounds and $\Omega(n)$ moves.
\end{itemize}
Next, when the cactus graph has at most $k$ ($k>1$) dynamic edge at any round.
\begin{itemize}
    \item Establish the impossibility to find the black hole with $k+1$ agents.
    \item With $k+2$ agents we establish lower bound of $\Omega({(n+2-3k)}^{1.5})$ rounds and $\Omega({(n+2-3k)}^{1.5}+2k)$ moves.
    \item With $2k+3$ agents improved lower bounds are $\Omega(n+2-3k)$ rounds, $\Omega(n+2-k)$ moves, and we establish an upper bound of $O(kn)$ rounds and $O(k^{2}n)$ moves.
\end{itemize}

\begin{table}[H]
\centering
\begin{tabular}[t]{|c|c|c|c|c|}
\hline

\# DE & \# Agents & Moves & Rounds & \\
\hline

1&3&$\Omega(n^{1.5})$&$\Omega(n^{1.5})$ & LB (Cor \ref{lowerbound-single-edge} \& Thm \ref{lowerbound-move-round-single-edge})\\\cline{2-5}

&3&$O(n^2)$&$O(n^2)$ & UB (Thm \ref{complexity1edge})\\\cline{2-5}

&4&$\Omega(n)$& $\Omega(n)$ & LB (Thm \ref{fouragent-LB})\\

\hline

$k$ &$k+2$& $\Omega({(n+2-3k)}^{1.5}+2k)$& $\Omega({(n+2-3k)}^{1.5})$ &LB (Cor \ref{lowerboundkedgecorollary} \& Thm \ref{k+2-LB})\\\cline{2-5}

& $2k+3$ &  $\Omega(n+2-k)$  & $\Omega(n+2-3k)$ & LB (Thm \ref{2k+3-LB})\\\cline{2-5}

 &$2k+3$&$O(k^2n)$& $O(kn)$& UB (Thm \ref{agentkedge} \& Thm \ref{multi-complexity})\\
\hline
\end{tabular}
\caption{Summary of Results $(k>1)$, where $LB$, $UB$ and $DE$ represent lower bound, upper bound and dynamic edge, respectively.}
\end{table}

\noindent\textbf{Organization:} Rest of the paper is organized as follows. In section \ref{Model}, we discuss the model and preliminaries. Section \ref{LB-results}, we give the lower bounds. In section \ref{BHS-section}, we present the algorithm and its correctness for both the single and multiple dynamic edge cases and finally concluding in section \ref{conclusion}. 

\section{Model and Preliminaries}\label{Model}

\noindent \textbf{Dynamic Graph Model:} We adapt the synchronous dynamic network model by Kuhn et al. \cite{kuhn2010distributed} to define our dynamic cactus graph $\mathcal{G}$. The vertices (or nodes) in $\mathcal{G}$ are static, whereas the edges are dynamic i.e., the edges can disappear (or in other terms go missing) and reappear at any round. The dynamicity of the edges holds as long as the graph is connected. The dynamic cactus graph $\mathcal{G}=(V,\mathcal{E})$ is defined as a collection of undirected cactus graphs $<G_0,G_1,\cdots,G_r,\cdots>$, where $G_r=(V,E_r)$ is the graph at round $r$, $|V|=n$ and $\mathcal{E}=\cup^{\infty}_{r=0}E_r$, where $|E_r|=m_r$ denotes the number of edges in $G_r$. The adversary maintains the dynamicity of $\mathcal{G}$, by disappearing or reappearing certain edges at any round $r$ such that the underlying graph is connected. This model of dynamic networks is studied in \cite{kuhn2010distributed} and is termed as a 1-interval connected network. The degree of a node $u\in \mathcal{G}$ is denoted by $deg(u)$, in other words, $deg(u)$ denotes the degree of the node $u$ in $G_0$. The maximum degree of the graph $\mathcal{G}$ is denoted as $\Delta$. The vertices (or nodes) in $\mathcal{G}$ are anonymous, i.e., they are unlabelled, although, the edges are labelled, an edge incident to $u$ is labelled via the port numbers $0,\cdots,deg(u)-1$. The ports are labelled in ascending order along the counter-clockwise direction, where a port with port number $i$ denotes the port number corresponding to the $i$-th incident edge at $u$ in the counter-clockwise direction. Any edge $e=(u,v)$ is labelled by two ports, one among them is incident to $u$ and the other incident to $v$, they have no relation in common (refer to Fig. \ref{cactus-desc-fig}). Any number of agents can pass through an edge concurrently. Each node in $\mathcal{G}$ has local storage in the form of a \textit{whiteboard}, where the size of the whiteboard at a node $v\in V$ is $O(deg(v)(\log deg(v)+k\log k))$, where $deg(v)$ is the degree of $v$ and $k$ is the atmost number of missing edge. The whiteboard is essential to store the list of port numbers attached to the node.  Any visiting agent can read and/or write travel information corresponding to port numbers. Fair mutual exclusion to all incoming agents restricts access to the whiteboard. The network $\mathcal{G}$ contains a malicious node termed as \textit{black hole}, which eliminates any incoming agent without leaving any trace of its existence.

\begin{figure}[h]
\centering
\includegraphics[width=0.5\linewidth]{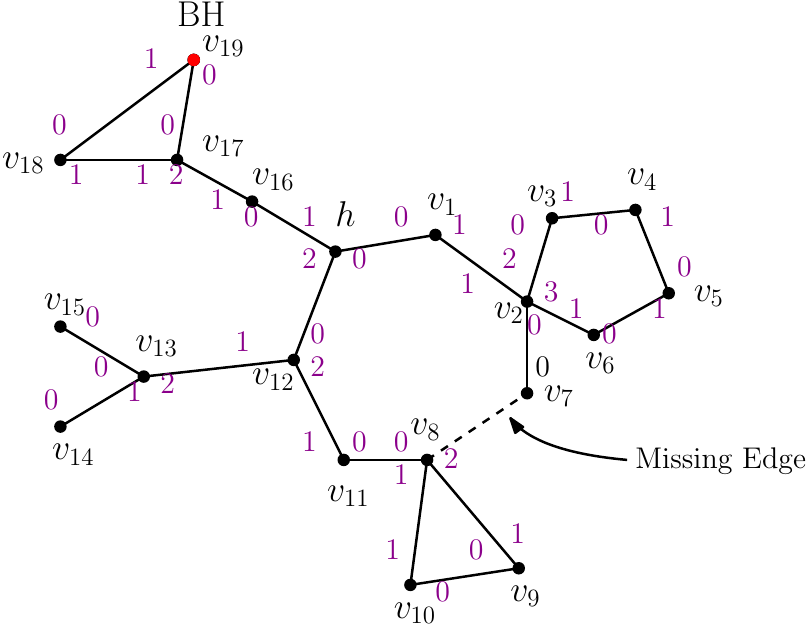}
\caption{A port labelled Cactus Graph is depicted where an edge $(v_7,v_8)$ is missing.}
\label{cactus-desc-fig}
\end{figure}

\noindent \textbf{Agent:} Let $\mathcal{A}$ = $\{a_1,\cdots, a_{m}\}$ be a set of $m\le n$ agents, they are initially co-located at a safe node termed as \textit{home}. Each agent has a distinct Id of size $\lfloor \log m \rfloor$ bits taken from the set $[1,m]$ where, each agent is a $t$-state automata, with local storage of $O(n\log \Delta)$ bits of memory, where $t \ge \alpha n \log \Delta$ and $\alpha$ is any positive integer. The agents visiting a node knows the degree of that node and also can determine the missing edges at that particular node, based on the whiteboard information. The agent while moving from $u$ to $v$ along the edge $e$ knows the port along which it left $u$ and the port along which it entered $v$. Further, the agents can see the $Id$s of other agents residing at the same node and can communicate with them.\\

\noindent \textbf{Round:} The agents operate in \textit{synchronous} rounds, where each agent gets activated in each round. At any time an agent $a_i\in \mathcal{A}$ gets activated, and they perform the following steps in a round: ``Communicate-Compute-Move" (CCM), while it is active. The steps are defined as follows:
\begin{itemize}
    \item \textit{Communication:} Agents can communicate among themselves when they are at the same node at the same time. They can communicate via whiteboard as well. In this step, agents can also observe their memory.
    \item \textit{Compute:} An agent based on the gathered information, local snapshot (i.e., information gathered on whether any other agent is present at the current node), internal memory and contents of a whiteboard, either decides to stay or choose the port number in case it decides to move. 
    
    \item \textit{Move:} The agent moves along the chosen port to a neighboring port, if it decides to move. While it starts to move, the agent writes the information in its memory and also writes on the whiteboard of its current node.
\end{itemize}
An agent takes one unit of time to move from a node $u$ to another node $v$ following the edge $e=(u,v)$.\\

\noindent\textbf{Time and Move Complexity:} Since the agents operate in synchronous rounds, each agent gets activated at each round to perform one CCM cycle synchronously. So, the time taken by the algorithm is measured in terms of \textit{rounds}. Another parameter is \textit{move} complexity, which counts the total number of moves performed by the agents during the execution of the algorithm.\\

\noindent\textbf{Configuration}: We define $C_r$ to be the \textit{configuration} at round $r$ which holds the following information: the contents in the whiteboard at each node, the contents of the memory of each agent and the locations of the agent at the start of round $r$. 

So, $C_0$ in the initial configuration at the start of an algorithm $\mathcal{H}$, whereas $C_r$ is the configuration obtained from $C_{r-1}$ after execution of the algorithm $\mathcal{H}$ at round $r$.

We recall next the notions of \textit{cautious}, \textit{pendulum} and \textit{pebble} walk.

\noindent{\bf Cautious Walk \cite{di2021black}.}
This is a movement strategy for agents in a network with a black hole, which ensures that at least two agents can travel together so that only one of them may be destroyed by the black hole, while the other survives. This is done as follows.
 If two agents $a_1$ and $a_2$ are at a safe node $u$, $a_1$ goes to an adjacent node $v$ and returns, while $a_2$ waits at $u$. 
If  $a_1$ returns to node $u$ at the next round, then both $a_1$ and $a_2$ can safely travel together to node $v$.
Otherwise, if at the next round, $a_1$ has not returned to node $u$, $a_2$ knows that either edge $(u,v)$ disappeared (which can be discovered by looking at the corresponding port) or that $v$ contains the black hole.\\

\noindent{\bf Pendulum Walk \cite{di2021black}.}
From a high-level perspective, an agent $a_1$ travels back and forth, increasing the number of hops at each movement, and always reports back to another agent $a_2$ (which may be referred to as a ``witness'' agent). 
More precisely, let us consider that two agents $a_1$ and $a_2$ are located at a node $u$. Now, $a_1$ decides to move one hop along the edge $(u,v)$ and reaches a node $v$. If $v$ is safe, then $a_1$ returns back to $u$ to inform $a_2$ that $v$ is safe. Next, $a_1$ decides to move two hops along the edge $(u,v)$ and $(v,w)$, thus reaching node $w$. If $w$ is safe, then $a_1$ returns back to $u$ via $v$. In general, $a_1$ in each movement increments the hop count by one, and at some round reaches a node $z$ along the path $P=u\rightarrow v \rightarrow w\rightarrow \cdots \rightarrow z$. If $z$ is safe, then it returns to $u$.\\

\noindent{\bf Pebble Walk:} This walk is a special case of \textit{pendulum} walk. In this case, as well an agent $a_1$ travels to and fro, but unlike \textit{pendulum} walk whenever $a_1$ reaches a new node it does not report back to another agent (or witness agent).

More precisely, let us consider an agent $a_1$ is currently at a node $u$. Now, $a_1$ decides to move one hop to an adjacent node $v$ using the edge $(u,v)$. If $v$ is safe, then it returns back to $u$. Further, $a_1$ again reaches $v$ and decides to move another one hop from $v$ to a new node $w$, following a similar strategy. So, the path followed by $a_1$ to reach $w$ is denoted as $P=u\rightarrow v \rightarrow u \rightarrow v \rightarrow w$. In general, $a_1$ moves one hop to a new node by similar to and fro movement from the last explored node.\\

\section{Lower Bound Results}\label{LB-results}

In this section, we first study the lower bound on the number of agents, move and round complexities required to solve the BHS problem when at most one edge can disappear or appear in a round. Then we generalize this idea for the case when at most $k$ edges can disappear or appear in a round such that the underlying graph remains connected.

\subsection{Lower Bound results on Single Dynamic Edge}

Here, we present all the results related to a dynamic cactus graph when at most one edge missing at any round.
\begin{theorem}[Impossibility for single dynamic edge]\label{lowerboundagent}
Given a dynamic cactus graph $\mathcal{G}$ of size $n>3$ with at most one edge which can disappear at any round such that the underlying graph is connected. Let the agents have the knowledge that the black hole is located in any of the three consecutive nodes $S=\{v_1,v_2,v_3\}$ inside a cycle of $\mathcal{G}$. Then it is not possible for two agents to successfully locate the black hole position. The impossibility holds even if the nodes are equipped with a whiteboard.  
\end{theorem}
The above theorem is a consequence of Lemma 1 in \cite{di2021black}.

\begin{corollary}[Lower bound for single dynamic edge]\label{lowerbound-single-edge}
To locate the black hole in a dynamic cactus graph $\mathcal{G}$ with at most one edge missing at any round, any algorithm requires at least 3 agents to solve the black hole search problem in $\mathcal{G}$.
\end{corollary}

\begin{lemma}[\cite{di2021black}]\label{expansion}
    If an algorithm solves black hole search with $O(n\cdot f(n))$ moves with three agents, then there exists an agent that explores a sequence of at least $\Omega(\frac{n}{f(n)})$ nodes such that:
    \begin{itemize}
        \item The agent does not communicate with any other node while exploring any node in the sequence.
        \item The agent visits at most $\frac{n}{4}$ nodes outside the sequence while exploring any node in the sequence.
    \end{itemize}
    This lemma holds even if the nodes are equipped with whiteboards.
\end{lemma}

In the next theorem, we give a lower bound on the move and round complexity required by any algorithm in order to solve the black hole search problem in a dynamic cactus.
\begin{theorem} \label{lowerbound-move-round-single-edge}
    Given a dynamic cactus graph $\mathcal{G}$, with at most one missing edge at any round. In the presence of a whiteboard, any algorithm $\mathcal{H}$ solves the black hole search problem with three agents in $\Omega(n^{1.5})$ rounds and $\Omega(n^{1.5})$ moves, when the agents have distinct IDs and they are co-located.
\end{theorem}

The above theorem is a consequence of theorem 6 in \cite{di2021black}. The next theorem, gives an improved lower bound on the move and round complexity when 4 agents try to locate the black hole instead of 3.

\begin{theorem}\label{fouragent-LB}
     Given a dynamic cactus graph $\mathcal{G}$, with at most one dynamic edge at any round. In the presence of whiteboard, any algorithm $\mathcal{H}$ solves the black hole search problem with four agents in $\Omega(n)$ rounds and $\Omega(n)$ moves, when the agents have distinct Ids and they are co-located.
\end{theorem}

\begin{proof}
     Let $a_1$, $a_2$, $a_3$ and $a_4$ are the four agents trying to find the black hole in $\mathcal{G}$. Consider $C_k$ be a cycle in $\mathcal{G}$ (refer to Fig. \ref{cactus-LB-fig}) and $y_k$ be the black hole node. Now, suppose $a_1$ enters the black hole and $Q$ be the set of consecutive nodes along which $a_1$, before entering the black hole has written its exact location on the whiteboard, so whenever an agent visits any of these nodes, it knows the exact location of the black hole and terminates. Let $e_q$ be the edge separating the black hole and $Q$ with the rest of the graph. Now, the adversary has the ability to restrict any agent from visiting the set of nodes in $Q$ by removing $e_q$. On the contrary, at any round, if no agent is trying to visit a node in $Q$, whereas establishing the black hole location along the counter-clockwise direction, then the adversary can restrict the agent visiting such a node by removing an edge, such as $e_{cc}$. So, the only possibility is while an agent always tries to visit a node in $Q$, the remaining two agents can correctly locate the black hole location while traversing a counter-clockwise direction in $C_k$, in at least $n$ rounds. Moreover, since at any round, a constant number of agents are moving. Hence, in order to successfully locate the black hole location with 4 agents, any algorithm requires $\Omega(n)$ rounds and $\Omega(n)$ moves. \qed\end{proof}

The next observation gives a brief idea about the movement of the agents on a cycle inside a dynamic cactus graph. It states that, when a single agent is trying to explore any unexplored cycle, the adversary has the power to confine the agent on any single edge of the cycle. Moreover, in case of multiple agents trying to explore a cycle inside a cactus graph, but their movement is along one direction, i.e., either clockwise or counter-clockwise, then also the adversary has the power to prevent the team of agents from visiting further unexplored nodes.

\begin{observation}
    Given a dynamic cactus graph $\mathcal{G}$, and a cut $U$ (with $|U|>1$) of its footprint connected by edges $e_1$ in the clockwise direction and $e_2$ in the counter clockwise direction to the nodes in $V \symbol{92} U$. If we assume that all the agents at round $r$ are at $U$, and there is no agent which tries to cross to $V \symbol{92} U$ along $e_1$ and an agent tries to cross along $e_2$, then the adversary may prevent agents to visit nodes outside $U$.
\end{observation}

The above observation follows from observation 1 of \cite{di2021black}. The next lemma follows from the structural property of a cactus graph.

\begin{figure}
\centering
\begin{minipage}{.4\textwidth}
  \centering
  \includegraphics[width=1.17\linewidth]{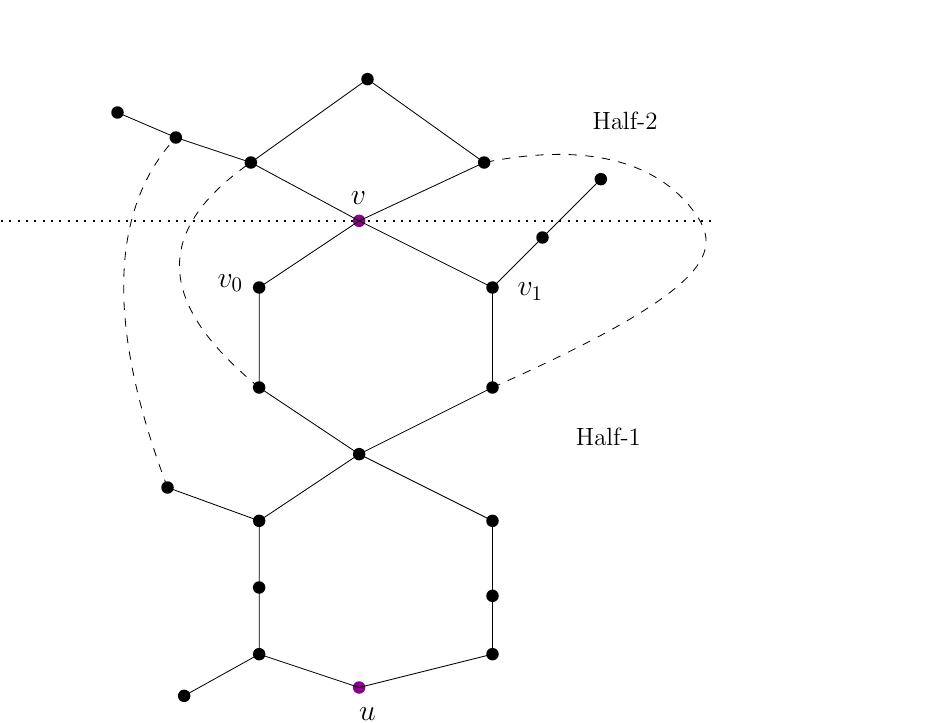}
\caption{Represents that any $u$ to $v$ path either passes through $v_0$ or $v_1$.}
\label{Cactus_path}
\end{minipage}%
\hspace{0.2cm}
\begin{minipage}{.5\textwidth}
  \centering
 \includegraphics[width=0.7\textwidth]{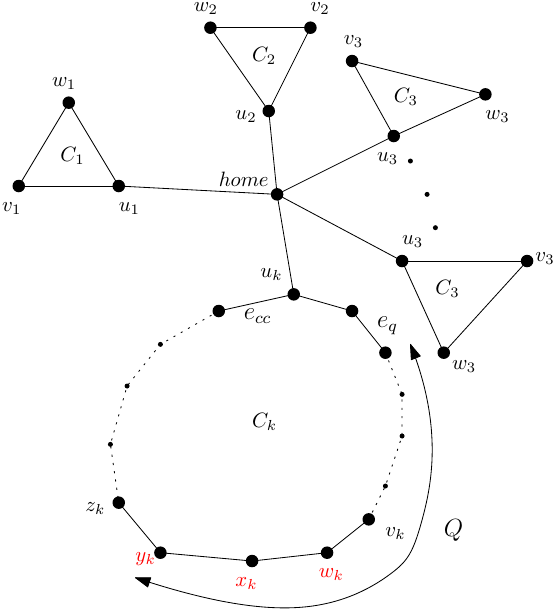}
\caption{A cactus graph, with $k$ different cycles, where red nodes in $C_k$ depicts possible location of the black hole.}
\label{cactus-LB-fig}
\end{minipage}
\end{figure}

\begin{lemma}\label{uv-path}
    Consider three consecutive nodes $\{v_0,v,v_1\}$ in a cactus graph $\mathcal{G}$, then any path from $u$ to $v$ in $\mathcal{G}$ must pass through either $v_0$ or $v_1$.
\end{lemma}
\begin{proof}
    We prove the above claim by contradiction. Suppose there exists a $u$ to $v$ path which neither passes through $v_0$ nor $v_1$, so in order to have an alternate path which does not pass through $v_0$ or $v_1$, implies that there must be at least one edge or a path passing from half-1 to half-2 (refer to Fig. \ref{Cactus_path}), where we define half-1 to be the subgraph on and above the horizontal half-line passing through $v$, whereas half-2 is the subgraph below the horizontal half-line passing through $v$. Now, the presence of such an edge or path, implies that there is at most one common edge between two cycles, and this violates the characteristic of a cactus graph. Hence, there cannot be any such $u$ to $v$ path which neither passes through $v_0$ nor $v_1$.
\qed \end{proof}

\subsection{Lower bound results for multiple dynamic edges}

Here, we present all the lower bound results for a dynamic cactus graph $G$, when at most $k$ edges are missing at any round.

\begin{theorem}[Impossibility for multiple dynamic edges]\label{lowerboundkedge}
    Given, a dynamic cactus graph $\mathcal{G}$ with at most $k$ dynamic edges at any round. It is impossible for $k+1$ agents, to successfully locate the black hole position, regardless of the knowledge of $n$, $k$ and the presence of a whiteboard at each node.
\end{theorem}
\begin{proof}
    We prove the above statement by contradiction. Let us consider a dynamic cactus graph $\mathcal{G}$ (refer to Fig. \ref{cactus-LB-fig}) of $n$ vertices, in which at most $k$ edges are dynamic at any round. The graph $\mathcal{G}$ contains $C_i$ cycles, where $i \in \{1,\cdots,k\}$, in which except the last cycle $C_k$ which of length $n+2-3k$, every other cycle is of length 3. Now, suppose there exists an algorithm $\mathcal{H}$, which successfully locates the black hole position in $\mathcal{G}$ with a set of $k+1$ agents, $\{a_1,\cdots, a_{k+1}\}$. Each agent is initially located at $home$, and after following the algorithm $\mathcal{H}$, the agents enter a configuration, where an agent $a_i$ reaches a node $v_i$ or $w_i$ inside $C_i$, where $i \in \{1,2,\cdots,k-1\}$ and the remaining two agents enter the cycle $C_k$. Suppose, the black hole is located at any one among the three consecutive nodes $S=\{w_{k},x_{k},y_{k}\}$, of which the agents have no idea. Since the adversary can disappear and reappear at most $k$ edges at any round, hence, it has the ability to restrict each $a_i$ inside $C_i$ by alternating its position between $v_i$ and $w_i$, by removing the edge $(u_i,v_i)$ and $(u_i,w_i)$ alternatively (where $i\in \{1,\cdots,k-1\}$). So, the remaining agents $a_k$ and $a_{k+1}$ have no other choice but to explore $C_k$. They cannot explore a node in $S$ together for the first time, as the adversary may place the black hole, eliminating both of them, whereas the other agents have no idea of the location of the black hole as they are unable to come out of $C_i$ ($i\in\{1,\cdots,k-1\}$). Let us consider, without loss of generality $a_{k}$ is the first to enter $S$, i.e., a node $w_k$ at some round $r$ or both agents enter a node in $S$ at round $r$ ($a_k$ enters $w_k$ and $a_{k+1}$ enters $y_k$). At this point, regardless of the position of $a_{k+1}$, the adversary removes the edge $(v_k,w_k)$. Now, in any case, $a_k$ cannot communicate with $a_{k+1}$ in the presence of a whiteboard as well, as on one side there is a black hole and on the other side, there is a missing edge. So, if $w_k$ is a safe node, then $a_1$ has no other option but to visit $x_k$ at some round. In the meantime, each of the remaining agents $a_i$ are stuck inside $C_i$ ($i\in\{1,\cdots,k-1\}$), respectively. Suppose if $y_k$ is safe, then $a_{k+1}$ can convey this information to at least one of the agents $a_i$ by writing this information in the whiteboard of either $u_i$ or $w_i$ (where $1\le i \le k-1$). So, now the only possibility for $a_{k+1}$ is to visit the next node $x_k$, if $x_k$ is the black hole, then the remaining $k-1$ agents cannot identify among $x_k$ and $w_k$ which is indeed the black hole node. The reason is, if $w_k$ is the black hole, then $a_k$ has already been eliminated, in this situation, the adversary can remove the edge $(y_k,x_k)$, which in turn restricts $a_{k+1}$ to come out of $C_k$ and convey this information to remaining $k-1$ agents. Now, since none of the remaining agents can come out of their respective cycles, regardless of the whiteboard information, left by the agents $a_{k}$ and $a_{k+1}$ before entering the black hole node, none of the remaining agents cannot figure out the exact location among $x_k$ and $w_k$ which is the correct black hole position. This leads to a contradiction.  
\qed\end{proof}

\begin{corollary}[Lower bound for $k$ dynamic edges]\label{lowerboundkedgecorollary}
To locate the black hole in a dynamic cactus graph $\mathcal{G}$ with at most $k$ dynamic edges at any round, any algorithm requires at least $k+2$ agents to solve the black hole search problem in $\mathcal{G}$.
\end{corollary}
 The next two theorems give lower bound and improved lower bound complexity with $k+2$ and $2k+3$ agents, respectively.

\begin{theorem}\label{k+2-LB}
    Given a dynamic cactus graph $\mathcal{G}$, with at most $k$ dynamic edges at any round. In presence of whiteboard, any algorithm $\mathcal{H}$ which solves the black hole search problem with $k+2$ agents requires $\Omega({(n+2-3k)}^{1.5})$ rounds and $\Omega({(n+2-3k)}^{1.5}+2k)$ moves when the agents have distinct $Id$s and they are co-located. 
\end{theorem}
\begin{proof}
    We prove the above statement by contradiction, we suppose that there exists an algorithm $\mathcal{H}$ which finds the black hole in $o({(n+2-3k)}^{1.5})$ rounds. Now, consider the graph $\mathcal{G}$ (refer to Fig. \ref{cactus-LB-fig}) of $k$-cycles, where $|C_i|=3$ ($1\le i \le k-1$) and $|C_k|=n+2-3k$, in which the set of agents $\mathcal{A}=\{a_1,a_2,\cdots,a_{k+2}\}$ are initially co-located at $home$. Suppose, while executing the algorithm $\mathcal{H}$, they enter a configuration, in which $a_i$ gets stuck inside $C_i$ ($1\le i \le k-1$), whereas $a_{k}$, $a_{k+1}$ and $a_{k+2}$ enter $C_k$. Now, by theorem 6 in \cite{di2021black}, a set of 3 agents requires $\Omega({(n+2-3k)}^{1.5})$ rounds to correctly locate the black hole inside $C_k$, this leads to a contradiction. Moreover, a constant number of agents move in $\Omega({(n+2-3k)}^{1.5})$ rounds, whereas to enter this configuration starting from $home$, at least $2k$ moves are required. Hence, the team of $k+2$ agents require $\Omega({(n+2-3k)}^{1.5})$ rounds and $\Omega({(n+2-3k)}^{1.5}+2k)$ moves to find the black hole.\qed\end{proof} 

\begin{theorem}\label{2k+3-LB}
    Any algorithm $\mathcal{H}$ in the presence of whiteboard which solves the black hole search problem with $2k+3$ agents in a dynamic cactus graph with at most $k$ dynamic edges at any round requires $\Omega(n+2-3k)$ rounds and $\Omega(n+2-k)$ moves when the agents have distinct Ids and they are co-located.
\end{theorem}

\begin{proof}
    We prove the above statement by contradiction, we suppose $\mathcal{H}$ finds the black hole in $o(n+2-3k)$ rounds. Now, suppose $\mathcal{G}$ be the graph (refer to Fig. \ref{cactus-LB-fig}) with $k$-cycles, where $|C_i|=3,~\forall ~1\le i \le k-1$ and $|C_k|=n+2-3k$, and consider the agents to be initially co-located at $home$. Now, suppose the agents while executing $\mathcal{H}$ enter a configuration, where $a_i$ gets stuck in $C_i$, where $1\le i \le k-1$. In this situation, the remaining $k+4$ agents try to explore $C_k$, so by theorem \ref{fouragent-LB} we know that it takes four agents among the $k+4$ agents to successfully locate the black hole in $\Omega(n+2-3k)$ rounds. The bound on the number of moves comes from the fact that at least $2k$ additional moves are required to attain this configuration, whereas a constant number of agents moves in $\Omega(n+2-3k)$ rounds. Hence, any algorithm requires $\Omega(n+2-3k)$ rounds and $\Omega(n+2-k)$ moves. \qed\end{proof}

\section{Black Hole Search in Dynamic Cactus}\label{BHS-section}
 In this section, we first present an algorithm to find the black hole in the presence of at most one dynamic edge, and then we present an algorithm to find the black hole in the presence of at most $k$ dynamic edges. The number of agents required to find the black hole is presented and further, the move and round complexities are analyzed for each algorithm. 

 \subsection{Black Hole Search in Presence of Single Dynamic Edge}\label{descsingleedge}
In this section, we present two black hole search algorithms, one for the agents and the other for the $LEADER$, in the presence of at most one dynamic edge in a cactus graph. Our algorithms require each node $v$ to have some local storage space, called \textit{whiteboard}. Our algorithms find the black hole search problem with the help of 3 agents $a_1$, $a_2$ and $a_3$, respectively. Among these three agents, we consider $a_3$ to be the $LEADER$.
The task of the $LEADER$ is different from the other two agents. The fundamental task of $LEADER$ is either to instruct an agent to perform a certain movement or to conclude that a certain agent has entered the black hole. In contrast, the agent's task is to visit nodes that are not explored. More precisely, our algorithm for $LEADER$ is {\textsc{SingleEdgeBHSLeader}} and for the agents $a_1$ and $a_2$ is {\textsc{SingleEdgeBHSAgent}}. Next, we discuss the contents of the whiteboard.

\noindent\textbf{Whiteboard:} For each node $v\in\mathcal{G}$, a whiteboard is maintained with a list of information for each port of $v$. For each $j$, where $j\in\{0,\cdots,deg(v)-1\}$, an ordered tuple  $(f(j),Last.LEADER)$ is stored on it, where the function $f$ is defined as follows: $f:\{0,\cdots,{deg(v)-1}\}\longrightarrow \{\bot,0,1\}^{*}$, 
\[f(j) = 
     \begin{cases}
       \text{$\bot$,} &\quad\text{if $j$ is yet to be explored by any agent}\\
       \text{$0\circ A$,} &\quad\text{if the set of agents in $A$ has visited $j$ but cannot}\\
       &\quad\text{fully explore the sub-graph originating from $j$}\\ 
       
       \text{1,} &\quad\text{if the sub-graph corresponding to $j$ is fully explored}\\ 
       &\quad\text{and no agent is stuck}\\
     \end{cases}\]
The symbol $`\circ '$ refers to the concatenation of two binary strings. We define $A$ to be the set of agents which has visited that particular port. More precisely, if $a_1$ and $a_2$ both visits the port $j$, then we have $A=\{a_1,a_2\}$. We discuss the entries on the whiteboard with the help of the following example. Consider a port $j$ at a node $u$, along which only $a_1$ has passed, but is unable to completely explore the sub-graph originating from $j$. In this case, the function $f(j)$ returns the binary string $001$, where the first $0$ represents that the sub-graph originating from $j$ is not fully explored, and the next $01$ represents the $Id$ of $a_1$, so we have $A=\{a_1\}$.

The entry $Last.LEADER$ stores the bit 1 if $j$ is the last visited port in $v$ by the $LEADER$, otherwise, it stores 0.

Each agent (i.e., $a_1$ and $a_2$) performs a $t$-\textsc{Increasing-DFS} \cite{fraigniaud2005graph}, where the movement $a_1$ and $a_2$ can be divided in to two categories \textit{explore} and \textit{trace}:
\begin{itemize}
    \item In \textit{explore}, an agent performs either \textit{cautious} walk or \textit{pendulum} walk depending on the instruction of the $LEADER$. In this case, an agent visits a node for the first time, i.e., it only chooses a port $j$, such that $f(j)=\bot$.
    \item In \textit{trace}, an agent performs \textit{pendulum} walk, where it visits a node that has already been visited by the other agent. In this case, an agent $a_1$ (or $a_2$) chooses a port $j$, where $f(j)=0\circ A$ and $A=\{a_2\}$ (or $A=\{a_1\}$). 
\end{itemize}

\noindent The task of the $LEADER$ can be explained as follows:
\begin{itemize}
    \item Instructs $a_1$ or $a_2$ to perform either \textit{cautious} or \textit{pendulum} walk.
    \item Maintains the variables ${length}_{a_i}$ and $P_{a_i}$ ($a_i\in\{1,2\}$), while $a_i$ is performing \textit{pendulum} walk. The $LEADER$ increments ${length}_{a_i}$ by 1, whenever $a_i$ traverses a new node and report this information to $LEADER$. Moreover, $P_{a_i}$ is the sequence port traversed by $a_i$ from the initial node from which it started its \textit{pendulum} walk to the current explored node.
    \item It also maintains the variables $L_{a_i}$ and $PL_{a_i}$, where $L_{a_i}$ calculates the length of the path traversed by $LEADER$ away from the initial position of $a_i$ and $PL_{a_i}$ stores the sequence of these ports.
    \item Lastly, terminates the algorithm whenever it knows the black hole position.
\end{itemize}

\noindent An agent $a_1$ or $a_2$ \textit{fails to report} to $LEADER$, for one of the following reasons: it has either entered the black hole or it has encountered a missing edge along its forward movement.

The algorithm \textsc{SingleEdgeBHSLeader} assigns only the $LEADER$ to communicate with the remaining agents, in order to instruct them regarding their movements, whereas the agents $a_1$ and $a_2$ do not communicate among themselves. The agents $a_i$ (where $i \in \{1,2\}$) stores $P_{a_i}$ while performing \textit{pendulum} walk. For example, let $a_1$ and $LEADER$ are co-located at $h$ in Fig. \ref{cactus-desc-fig}. Now, suppose $a_1$ starts \textit{pendulum} walk and reaches $v_1$. Correspondingly, $P_{a_2}$ gets updated to $P_{a_2}\cup 0$.

Next, we discuss in detail the description of the algorithm.

\noindent \textbf{Algorithm Description:} 
In this section, we give a high-level description of the algorithms \textsc{SingleEdgeBHSAgent} and \textsc{SingleEdgeBHSLeader}.\\

\noindent\textit{Description of \textsc{SingleEdgeBHSAgent()}}: This algorithm is executed by the agents $a_1$ and $a_2$, in which they are either instructed to perform \textit{cautious} or \textit{pendulum} walk. The agents perform $t$-\textsc{Increasing-DFS} \cite{fraigniaud2005graph} strategy for deciding the next port, further that port is indeed chosen by the agent for its movement based on the whiteboard information. Before commencing the algorithm, each node is labelled as $(\bot,0)$. Initially, the agents start from $home$, where without loss of generality, $a_1$ is instructed to perform \textit{cautious} and $a_2$ is instructed to perform \textit{pendulum} walk. Let us consider $a_1$ is at a node $u$, then the decision taken by $a_1$ based on the contents of the whiteboard is as follows:
\begin{itemize}
    \item If $\exists$ at least one port with $f()$ value is $\bot$ and $i$ being the minimum among them, then choose that port and move to its adjacent port $u'$ via $i$ from $u$.  
    \item If there is no such port $i$ at $u$, with $f(i)=\bot$, then it backtracks accompanying the $LEADER$ to a node where $\exists$ such $i$ with $f(i)=\bot$.
\end{itemize}
Now, suppose $a_1$ reaches $u'$ through the $i$-th port and it is safe, then it returns back to $u$ in the subsequent round, on condition that the edge $(u,u')$ remains. Otherwise, it stays at $u'$ until the edge reappears. Now, while it returns back to $u$, if $LEADER$ is found, then it accompanies $LEADER$ to $u'$ in the next round. Otherwise, it moves towards $LEADER$ following $Last.LEADER$ entries on whiteboard.

Next, let us consider $a_2$ is at a node $v$, then the decision taken by $a_2$ based on the whiteboard contents at $v$ is as follows:

\begin{itemize}
    \item If $\exists$ at least one port with $f()$ value is $\bot$ and $i$ being the minimum among them, then choose that port and move to its adjacent port $v'$ via $i$ from $v$.
    \item If there is no such port $i$ with $f(i)=\bot$, but $\exists$ a port $j$ with $f(j)=0\circ A$, where $A=\{a_1\}$ then it chooses that port and moves to its adjacent node $v'$. Otherwise, if $A=\{a_2\}$ or $A=\{a_1,a_2\}$, then it chooses a different available port, or backtracks to a node where there exists an available.
    \item If all ports have the value $1$, then it backtracks to a port where each port value is not $1$. 
\end{itemize}
Suppose, $a_2$ travels to $v'$ using the port $i$, then first it stores the port $i$, and after visiting $v'$ it moves towards $LEADER$, based on the stored ports, if it is unable to find $LEADER$, then it follows $Last.LEADER$ to meet the $LEADER$. Now, whenever it meets the $LEADER$ it provides the sequence of ports $P_{a_i}$ it has traversed from its initial position to the $LEADER$. Moreover, if at any moment it encounters a missing edge and no other agent is waiting for that missing edge, then it waits until the edge reappears. Now, irrespective of \textit{cautious} or \textit{pendulum} walk, whenever an agent $a_1$ (or $a_2$) moves along a port $i$ (say), it updates $f(i)=f(i)\circ a_1$ (or $f(i)=f(i)\circ a_2$) in whiteboard with respect to $i$.\\

\noindent\textit{Description of \textsc{SingleEdgeBHSLeader()}}: This algorithm is executed by the $LEADER$. It initially instructs $a_1$ to perform \textit{cautious} and $a_2$ to perform \textit{pendulum} walk. Whenever an agent, suppose $a_2$ while performing \textit{pendulum} walk, explores a new node and meets with $LEADER$, it increments ${length}_{a_2}$ by 1 and stores the sequence of port $P_{a_i}$ from $a_i$. On the other hand, if the $LEADER$ moves from its current position away from $a_2$, it increments $L_{a_2}$ by 1 at each such movement while updating the sequence of ports $PL_{a_i}$ after each such movement. Suppose, if the $LEADER$ moves away from $a_i$ from its current node $u$ to a node $v$ along the port $i$. It does the following things: first, it updates $Last.LEADER$ at $u$ corresponding to $i$ as 1, while the rest to 0. Second, it increments $L_{a_i}$ by 1 and lastly, updates $PL_{a_i}=PL_{a_i}\cup \{i\}$. Further, whenever $LEADER$ finds a missing edge along its path, it stops until the edge reappears. The instructions made by $LEADER$ related to the movement of $a_1$ and $a_2$ are as follows:
\begin{itemize}
    \item If $a_2$ without loss of generality (w.l.o.g), fails to report while performing \textit{pendulum}, while $a_1$ is performing \textit{cautious} walk, then $LEADER$ instructs $a_1$ to perform \textit{pendulum} walk.
    \item If $LEADER$ is stuck at one end of the missing edge, then it instructs both $a_1$ and $a_2$ to perform \textit{pendulum} walk if not already performing the same.
    \item If $LEADER$ finds a missing edge reappear, while both $a_1$ and $a_2$ performing \textit{pendulum} walk, then it instructs either $a_1$ or $a_2$ to perform \textit{cautious} walk, based on the fact which among $a_1$ or $a_2$ is faster to report to $LEADER$. 
\end{itemize}

\begin{figure}
 \centering
 \includegraphics[width=0.8\linewidth]{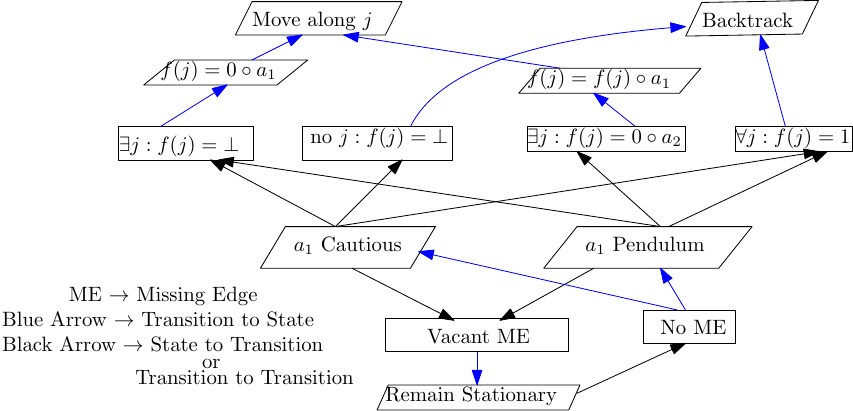}
 \caption{An illustration of algorithm \textsc{SingleEdgeBHSAgent()} in terms of the agent $a_1$}
 \label{StateDiagram-Agent}
 \end{figure}

Next, we discuss the situations, when $LEADER$ decides that either $a_1$ or $a_2$ has entered the black hole and terminates the algorithm.
\begin{itemize}
    \item If $a_1$ w.l.o.g, while performing \textit{cautious} walk fails to return, but the edge between $LEADER$ and $a_1$ remains, then $LEADER$ identifies $a_1$ to be in black hole.
    \item If $a_1$ w.l.o.g, while performing \textit{cautious} walk fails to return, but the edge between $LEADER$ and $a_1$ is missing, on the other hand $a_2$ which is performing \textit{pendulum} walk also fails to return, then $LEADER$ identifies $a_2$ to be in black hole.
    \item If both $a_1$ and $a_2$ is performing \textit{pendulum} walk, and $LEADER$ is stuck at one end of the missing edge. In this situation if both $a_1$ and $a_2$ fail to report, then $LEADER$ moves towards the agent $a_1$ (or $a_2$) with the help of $P_{a_1}$ and $PL_{a_1}$ (or $P_{a_2}$ and $PL_{a_2}$) based on the fact which among them is last to report to the $LEADER$, suppose that agent is $a_1$, then we have the following cases:
    \begin{itemize}
        \item If $a_1$ is found, then the $LEADER$ understands $a_2$ in black hole.
        \item If $a_1$ cannot be found and there is no missing edge, then $LEADER$ identifies $a_1$ to be in a black hole.
        \item If $a_1$ cannot be found, and there is a missing edge, then $LEADER$ waits, and if then also $a_2$ fails to report then $LEADER$ identifies $a_2$ to be in a black hole.
    \end{itemize}
\end{itemize}

Figure \ref{StateDiagram-Agent}, represents all possible states that an agent $a_1$ or $a_2$ attain, while executing \textsc{SingleEdgeBHSAgent()}, whereas Figure \ref{StateDiagram-Leader}, represents all possible states the $LEADER$ attains while executing \textsc{SingleEdgeBHSLeader()}. The pseudocode of \textsc{SingleEdgeBHSAgent()} is explained in Algorithm \ref{algorithmagent}, whereas the pseudocode of \textsc{SingleEdgeBHSLEADER()} is explained in Algorithm \ref{algorithmLeader}.

\begin{figure}
 \centering
 \includegraphics[width=0.9\linewidth]{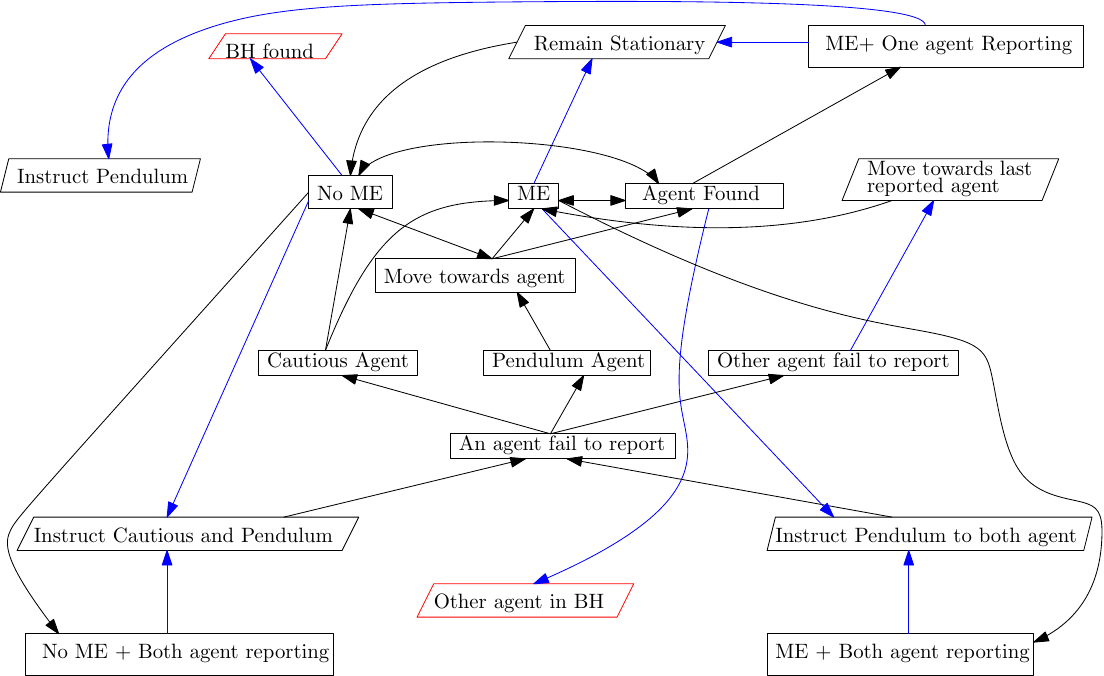}
 \caption{An illustration of algorithm \textsc{SingleEdgeBHSLeader()}}
 \label{StateDiagram-Leader}
 \end{figure}

\begin{algorithm2e}[!ht]\footnotesize
\caption{\sc{SingleEdgeBHSAgent($a_i$)}}\label{algorithmagent}
Set $P_{a_i}= \emptyset$.  \\ 
\While{$a_i$ performs cautious walk}
{
\If{each port at $u$ has $f(j)\neq \bot$}
{
Backtrack with $LEADER$ to a node $w$, where $\exists ~j$, such that $f(j)=\bot, ~j\in\{0,1,\cdots,deg(w)-1\}$.\\
}
\Else
{
Choose the minimum $j$ at $u$ with $f(j)=\bot$, and update $P_{a_i}=P_{a_i}\cup \{j\}$ and $f(j)=0\circ a_1$.\\ 
Travel to the adjacent node $v$ via $j$. \\
\If{$v$ is safe and $e=(u,v)$ remains}
{
Return back to $u$ and take $LEADER$ to $v$.\\
}
\ElseIf{$v$ is safe and $e=(u,v)$ is missing}
{
Remain at $v$ until $e$ reappears.\\
}
}
}
\While{$a_i$ performs pendulum walk}
{
\If{each port at $u$ has $f(j)=1$}
{
Backtrack to a node $w$ where $\exists ~j$ such that $f(j)\neq 1, ~j\in\{0,1,\cdots,deg(w)-1\}$. \\
\If{a missing edge is encountered while backtracking}
{
Remain stationary until the edge reappears.\\
}
}
\Else
{
Choose the minimum port $j$ such that $f(j)\neq 1$ or $0\circ A$, where $a_i \in A$. \\
Store $P_{a_i}=P_{a_i}\cup \{j\}$ and move to adjacent node $v$ via $j$.\\ 
Update $f(j)=f(j)\circ a_i$ and return to $LEADER$ following 
$P_{a_i}$ and $Last.LEADER$.\\
\If{a vacant missing edge is encountered}
{
Remain stationary until the edge reappears.\\
}
\Else
{
Choose a different available port or backtrack.\\
}
Return back to $v$ using $P_{a_i}$.\\
}
}
\end{algorithm2e}

\begin{algorithm2e}[!ht]\footnotesize
\caption{\sc{SingleEdgeBHSLeader($a_1,{length}_{a_2},L_{a_2},P_{a_2},PL_{a_2}$)}}\label{algorithmLeader}
Instruct $a_1$ to perform \textit{cautious} and $a_2$ to perform \textit{pendulum} walk.\\
Set ${length}_{a_1}=L_{a_1}=0$ and $PL_{a_1}=\emptyset$.\\
\While{the black hole is not found}
{
\If{$a_1$ does not return after a round} 
{
\If{no missing edge}
{
Conclude $a_1$ has entered the black hole and terminate.\\
}
\Else
{
Remain stationary until the edge reappears.\\
\If{$a_2$ fail to return within $2(L_{a_2}+{length}_{a_2}+1)$ round}
{
\If{the edge is still missing}
{
Conclude $a_2$ has entered the black hole and terminate.\\
}
\Else
{
Instruct $a_1$ to perform \textit{pendulum} walk and move towards $a_2$ following $P_{a_2}$ and $PL_{a_2}$ while each time incrementing $L_{a_1}$ by 1 and updating $PL_{a_1}$ by storing each port traversed.\\
Set $Last.LEADER=1$ for each port traversed and $Last.LEADER=0$ for rest of the ports at each node.\\
\textsc{Pendulum($a_2,{length}_{a_1},{length}_{a_2},L_{a_1},L_{a_2},P_{a_1},P_{a_2},PL_{a_1},PL_{a_2}$)}.\\
}
}
\Else
{
Update ${length}_{a_2}={length}_{a_2}+1$ and store $P_{a_2}$ from $a_2$.\\
}
}
}
\Else
{
Update $L_{a_2}=L_{a_2}+1$ and $PL_{a_2}=PL_{a_2}\cup \{j\}$, where $j$ is the port taken by $LEADER$.\\
\If{$a_2$ fail to return within $2(L_{a_2}+{length}_{a_2}+1)$ round}
{
Instruct $a_1$ to perform \textit{pendulum} walk and move towards $a_2$ following $P_{a_2}$ and $PL_{a_2}$ while each time incrementing $L_{a_1}$ by 1 and updating $PL_{a_1}$ by storing each port traversed.\\
Set $Last.LEADER=1$ for each port traversed and $Last.LEADER=0$ for rest of the ports at each node.\\
\textsc{Pendulum($a_2,{length}_{a_1},{length}_{a_2},L_{a_1},L_{a_2},P_{a_1},P_{a_2},PL_{a_1},PL_{a_2}$)}.\\
}
\Else
{
Update ${length}_{a_2}={length}_{a_2}+1$ and store $P_{a_2}$ from $a_2$.\\
}
}
}
\end{algorithm2e}

\begin{algorithm2e}[!ht]\footnotesize
\caption{\sc{Pendulum($a_2,{length}_{a_1},{length}_{a_2},L_{a_1},L_{a_2},P_{a_1},P_{a_2},PL_{a_1},PL_{a_2}$)}}\label{pendulum}
\While{black hole is not found}
{
\If{$a_2$ is found}
{
Set ${length}_{a_2}=L_{a_2}=0$ and $PL_{a_2}=\emptyset$.\\
\If{a forward edge is missing}
{
Remain stationary and instruct $a_2$ to perform \textit{pendulum} walk.\\
\If{$a_1$ does not return after $2(L_{a_1}+{length}_{a_1}+1)$ round}
{
Wait for $2(L_{a_2}+{length}_{a_2}+1)$ round.\\
\If{$a_2$ does not return to $LEADER$}
{
Traverse towards $a_2$ following $P_{a_2}$ and $PL_{a_2}$, while updating $Last.LEADER$ and incrementing $L_{a_1}$ and updating $PL_{a_1}$ for each port traversed.\\
\If{$a_2$ is found}
{
Conclude the $a_1$ has entered black hole and terminate.\\
}
\Else
{
If a forward edge is missing, then wait for $2({length}_{a_1}+1+L_{a_1})$ rounds.\\
If the edge is still missing and $a_1$ does not return, conclude $a_1$ is in black hole and terminate.\\
Otherwise, conclude $a_2$ has entered black hole and terminate.\\
}
}
\Else
{
Increment ${length}_{a_2}={length}_{a_2}+1$ and update $P_{a_2}$ from $a_2$.\\
\textsc{Pendulum($a_2,{length}_{a_1},{length}_{a_2},L_{a_1},L_{a_2},P_{a_1},P_{a_2},PL_{a_1},PL_{a_2}$)}.\\
}
}
\Else
{
Increment ${length}_{a_1}={length}_{a_1}+1$ and update $P_{a_1}$ from $a_1$.\\
}
}
\Else
{
\textsc{SingleEdgeBHSLeader($a_2,{length}_{a_1},L_{a_1},P_{a_1},PL_{a_1}$)}.\\
}
}
\ElseIf{it encounters a missing edge}
{
\textsc{Missing($a_2,{length}_{a_1},{length}_{a_2},L_{a_1},L_{a_2},PL_{a_1},P_{a_1},PL_{a_2},P_{a_2}$)}.\\
}
\Else
{
Wait for one round, and conclude $a_2$ has entered black hole and terminate.\\
}
}
\end{algorithm2e}

\begin{algorithm2e}[!ht]\footnotesize
\caption{\sc{Missing($a_2,{length}_{a_1},{length}_{a_2},L_{a_1},L_{a_2},PL_{a_1},P_{a_1},PL_{a_2},P_{a_2}$)}}\label{missing}
\While{$a_2$ is not found}
{
Remain stationary at one end of the missing edge.\\
\If{$a_1$ fails to report within $2(L_{a_1}+{length}_{a_1}+1)$ rounds}
{
Move towards $a_1$ following $P_{a_1}$ and $PL_{a_1}$ while updating $L_{a_2}$, $PL_{a_2}$ and $Last.LEADER$.\\
\If{$a_1$ is not found}
{
\If{there is no missing edge}
{
Conclude $a_1$ has entered black hole and terminate.\\
}
\Else
{
Remain stationary.\\
\If{$a_2$ does not report after $2(L_{a_2}+{length}_{a_2}+1)$ rounds}
{
Conclude $a_2$ has entered black hole and terminate.\\
}
\Else
{
Conclude $a_1$ has entered black hole and terminate.\\
}
}
}
}
\Else
{
Increment $length_{a_1}=length_{a_1}+1$ and update $P_{a_1}$ from $a_1$.\\
}
}
\If{the forward edge is missing}
{
\textsc{Pendulum($a_2,{length}_{a_1},{length}_{a_2},L_{a_1},L_{a_2},P_{a_1},P_{a_2},PL_{a_1},PL_{a_2}$)}.\\
}
\Else
{
\textsc{SingleEdgeBHSLeader($a_2,{length}_{a_1},L_{a_1},P_{a_1},PL_{a_1}$)}.\\
}

\end{algorithm2e}

\subsubsection{Correctness and Complexity:} In this section, we prove the correctness of our algorithm, as well as give the upper bound results in terms of move and round complexity.

\begin{lemma}\label{atmosttwo}
    Given a dynamic cactus graph $\mathcal{G}$ with at most one dynamic edge at any round $r$. Our Algorithms \textsc{SingleEdgeBHSAgent()} and \textsc{SingleEdgeBHSLeader()}, ensures that at most 2 agents are stuck due to the missing edge.
\end{lemma}
\begin{proof}

    In order to prove this claim, we first discuss the decisions that $LEADER$ takes following Algorithm \ref{algorithmLeader}, whenever it either finds an agent $a_1$ (w.l.o.g) or a missing edge.
    
    \begin{itemize}
        \item \underline{$LEADER$ finds $a_1$ at $u$ stuck due to a missing edge $(u,v)$}: In this case, if the edge does not reappear, then $LEADER$ instructs $a_1$ to perform \textit{pendulum} walk, whereas it remains stationary at $u$ until the edge $(u,v)$ reappears.
        \item \underline{$LEADER$ while tracing for $a_1$ finds a missing edge $(u,v)$}: If $a_1$ has not already entered the black hole, then it is at $v$ while the $LEADER$ is at $u$, because $a_1$ while reporting back to $LEADER$ gets stuck at $u$, whereas $LEADER$ after finding $a_1$ is not reporting back, moves towards $a_1$ and gets stuck at $u$. 
    \end{itemize}
    So, in any case, $LEADER$ does not allow more than one agent to occupy one end of the missing edge. 
    Next, we discuss the decisions taken by the agent $a_1$ (w.l.o.g) following the Algorithm \ref{algorithmagent}, whenever it either finds another agent or encounters a missing edge.
    \begin{itemize}
    \item \underline{If $a_1$ finds $u$ to be vacant and $(u,v)$ is a missing edge}: In this case, $a_1$ remains stationary at $u$, until either the edge reappears or until $LEADER$ appears at $u$.
    \item \underline{If the node $u$ of the missing edge $(u,v)$ is occupied by $a_2$}: In this case, if $a_1$ does not find any available ports at $u$ then it backtracks, otherwise, it moves along the available port at $u$.
    \end{itemize}
    So, irrespective of the situation from the above steps, we conclude that at most 2 agents can occupy a node on each side of the missing edge. All the above scenarios will hold as well if instead of $a_1$, the agent is $a_2$. This proves our statement.   
\qed\end{proof}
\begin{lemma}\label{atmosttwoLeader}
    The $LEADER$ following the algorithm \textsc{SingleEdgeBHSLeader()}, ensures that among the two agents stuck due to a dynamic edge, eventually, one must be the $LEADER$.
\end{lemma}
\begin{proof}
   We prove the above claim by contradiction. We claim that, whenever two agents $a_1$ and $a_2$ are at either side of a missing edge, they invariably need to hold these positions until the missing edge reappears. Consider the scenario, where an edge $(u,v)\in \mathcal{G}$ is missing for a finite but sufficiently large number of rounds, and $a_1$ and $a_2$ are stuck at $u$ and $v$, respectively. In the first case, we consider the scenario where both the agents are performing \textit{pendulum}. Now, this scenario has arised because $LEADER$ is occupying one end of an earlier missing edge, i.e., either at $u'$ or $v'$ (refer to Algorithm \ref{pendulum}). Since, at a round at most one edge can be missing, hence this earlier missing edge has reappeared, making way for the new missing edge $(u,v)$, which has stuck both $a_1$ and $a_2$ at $u$ and $v$, respectively. Now, as both of them fail to report, the $LEADER$ moves towards the agent which has last reported (say $a_1$ without loss of generality) and finds $a_1$. Whenever it finds $a_1$, it instructs $a_1$ to continue performing \textit{pendulum} walk while it remains stationary at $u$. This leads to a contradiction. In the second case, suppose $a_1$ at $u$ finds a missing edge $(u,v)$, while it is performing \textit{cautious} walk with $LEADER$, whereas $a_2$ at $v$ also finds a missing edge $(u,v)$ while it is performing \textit{pendulum} walk. In this scenario, $a_1$ and $LEADER$ are stuck at $u$ whereas $a_2$ is stuck at $v$. In this case, $LEADER$ following the Algorithm \ref{algorithmLeader} instructs $a_1$ to perform \textit{pendulum} walk, whereas it remains stationary at $u$. In this case, also, our claim leads to a contradiction. So we conclude that eventually among the two agents stuck due to a missing edge, one must be the $LEADER$. \qed\end{proof}

\begin{lemma}\label{infinite}
    An agent $a_1$ or $a_2$ following the \textsc{SingleEdgeBHSAgent()}, does not enter an infinite cycle.
\end{lemma}
\begin{proof}
    Suppose $C_1$ be a cycle in $\mathcal{G}$, where $u$ be the node along which an agent $a_1$ (w.l.o.g) enters the cycle and moves to the adjacent node $v$ using the port $j$. While it moves to $v$ via $j$, it updates $f(j)=f(j)\circ a_1$. Now, after exploring $C_1$, whenever the agent returns back to $u$ and again tries to take the port $j$, it finds that $a_1$ has already visited the port $j$ based on the whiteboard information, so irrespective of the agent performing \textit{cautious} or \textit{pendulum} walk, it does not take the port $j$ again (refer to steps 3-4 and 18 of Algorithm \ref{algorithmagent}). This phenomenon of an agent, concludes that it never enters an infinite cycle while performing \textsc{SingleEdgeBHSAgent()}. 
\qed\end{proof}

\begin{lemma}\label{explore}
    Algorithm \textsc{SingleEdgBHSAgent()}, ensure that in the worst case, every node in $\mathcal{G}$ is explored by either $a_1$ or $a_2$ until the black hole node is detected.
\end{lemma}
\begin{proof}
    Observe, each agent is a $t$-state finite automata, where $t \ge \alpha n\log\Delta$, the agents $a_1$ and $a_2$ performs $t-$\textsc{Increasing-DFS} while \textit{explore} or \textit{trace}. In addition, the whiteboard information helps $a_1$ and $a_2$ find the set of ports it hasn't visited yet at each node. Now, as stated in theorem 6 and corollary 7 in \cite{fraigniaud2005graph}, an agent with $O(n\log\Delta)$-bits of memory can explore any static graph with maximum degree $\Delta$ and diameter at most $n$. Moreover, lemma \ref{infinite} ensures that $a_1$ or $a_2$ never enters an infinite cycle. Further, lemma \ref{atmosttwoLeader} ensures that either $LEADER$ and one among $a_1$ or $a_2$ gets stuck by a missing edge, whereas the other agent can explore unobstructively until it either enters the black hole or detects it, as the underlying graph can have at most one dynamic edge. So, combining all these phenomenon, we can guarantee that $a_1$ or $a_2$, can explore the underlying graph until the black hole is detected.  \qed\end{proof}
    \begin{observation}\label{leader}
    Algorithm \textsc{SingleEdgeBHSLeader}, ensures that $LEADER$ does not enter the black hole.
\end{observation}

\begin{lemma}\label{agentBH}
In the worst case, $a_1$ and $a_2$ executing algorithm \textsc{SingleEdgBHSAgent()}, enters black hole in $O(n^2)$ rounds.
\end{lemma}
\begin{proof}
    Recall, $LEADER$ either instructs both $a_1$ and $a_2$ to perform \textit{pendulum} walk, or it instructs one among them to perform \textit{cautious} walk while the other \textit{pendulum} walk.
    \begin{itemize}
        \item \textit{\underline{$a_1$ is instructed cautious walk, whereas $a_2$ pendulum walk.}} Since, at most one edge is missing, then if $a_2$ is not blocked at all by any missing edge, then in at most $2n$ rounds, it explores a new node.
        \item \textit{\underline{Both $a_1$ and $a_2$ are instructed pendulum walk.}} This scenario arises when the $LEADER$ is stationary at one end of the missing edge. Now, in the worst case the missing edge for which $LEADER$ is stationary reappears, whereas both $a_1$ and $a_2$ while exploring a cycle along disjoint paths, gets blocked by a new missing edge (refer to Fig. \ref{cactus-desc-fig}, where suppose $a_1$ is at $v_7$ and $a_2$ at $v_8$). In this case, $LEADER$ finding no missing edge adjacent to itself, reaches either of them in at most $n$ rounds. Now, if $a_1$ is reached, then $LEADER$ instructs $a_1$ to continue \textit{pendulum} walk, while $LEADER$ remains stationary. In this scenario as well, $a_1$ explores a new node each in at most $2n$ rounds.
    \end{itemize}
    Hence, in any case, in at most $2n$ rounds, an agent explores at least one new node. This implies in $O(n)$ rounds, at least a new node is explored. This concludes the statement that either $a_1$ or $a_2$ enters the black hole in $O(n^2)$ rounds.  
\qed\end{proof}

\begin{lemma}\label{round}
Let us consider, the agents $a_1$ or $a_2$ enter the black hole at round $r$ while executing \textsc{SingleEdgeBHSAgent()}, then the $LEADER$ following \textsc{SingleEdgeBHSLeader()} detects the black hole node in $r+O(n^2)$ rounds.
\end{lemma}
\begin{proof}
    Suppose after $r$ rounds either $a_1$ or $a_2$ enters the black hole. Now in each of the following scenarios, we show that the $LEADER$ takes additional $O(n^2)$ rounds after round $r$ to detect the black hole:
    \begin{itemize}
        \item \textit{Scenario-1:} Suppose w.l.o.g, $a_1$ while performing \textit{cautious} walk with the $LEADER$ enters the black hole at round $r$. Hence, $a_1$ does not return back to $LEADER$, whereas the edge $e$ between the $LEADER$ and black hole is present. So, after waiting for one more round, $LEADER$ terminates the algorithm by concluding the black hole position. On the contrary, suppose the edge $e$ remains missing from the round $r$ onwards. Since, there is at most one missing edge at any round, and moreover the other end of the missing edge contains the black hole. So, $a_2$ can unobstructively explore a new node in at most $2n$ rounds while executing \textit{pendulum} walk, until it reaches the black hole. Hence, in $O(n^2)$ rounds, the agent $a_2$ enters black hole. Therefore, $a_2$ will also fail to report, and $LEADER$ will conclude the black hole position in $O(n)$ additional rounds from the round when $a_2$ enters the black hole. In conclusion, the $LEADER$ takes $O(n^2)$ rounds to conclude the black hole position, after $a_1$ has entered the black hole.
        \item \textit{Scenario-2:} Consider w.l.o.g, $a_2$ while performing \textit{pendulum} walk reaches the black hole at round $r$, while $a_1$ is performing \textit{cautious} walk with $LEADER$.
        So after $a_2$ fails to report, $LEADER$ performs the following task.
         \begin{itemize}
             \item If $a_1$ is with $LEADER$ when $a_2$ fails to report whereas the forward edge is missing, then $LEADER$ instructs $a_1$ to perform \textit{pendulum} walk whenever this edge reappears. Subsequently, $LEADER$ moves towards $a_2$.
             \item If $a_1$ is with $LEADER$ when $a_2$ fails to report and there is no missing edge in the forward direction, then $LEADER$ instructs $a_1$ to perform \textit{pendulum} walk, whereas $LEADER$ moves towards $a_2$.
         \end{itemize}
       So, in at most $n$ rounds, $LEADER$ reaches the last marked node of $a_2$ following $P_{a_2}$ and $PL_{a_2}$. Now, if $a_2$ is not found and there is no forward missing edge, then $LEADER$ concludes that $a_2$ has reached the black hole. This conclusion is executed in additional $O(n)$ rounds after $a_2$ has entered the black hole, i.e., after round $r$.
       
       Now on the contrary, if $a_2$ is not found by $LEADER$, and there exists a missing edge $(u,v)$ along the path towards $a_2$ in the forward direction, then $LEADER$ remains stationary at $u$ (say), but $a_1$ invariably performs \textit{pendulum} walk until it either reaches the black hole or gets stuck at $v$. The $LEADER$ in $O(n^2)$ rounds, understands the failure of $a_1$'s return as well, leaves $u$ and moves towards $a_1$. The possibilities are: either $a_1$ is found, or the $LEADER$ encounters another missing edge, or $a_2$ is not found. In each case, $LEADER$ concludes the black hole position in $O(n^2)$ rounds after $a_2$ has entered the black hole. 

       \item \textit{Scenario-3}: Both $a_1$ and $a_2$ are performing \textit{pendulum} walk, in which $a_1$ (say) enters the black hole. Now, again $a_2$ unobstructively explores the graph in $O(n^2)$ rounds while the $LEADER$ is stationary at one end of the missing edge, until $a_2$ either enters black hole or gets stuck by a missing edge. Now, the $LEADER$ invariably travels towards $a_2$, while it faces the following instances: either $a_2$ is found, or $a_2$ is not found, or encounters a missing edge. In any case, the $LEADER$ concludes the black hole position in $O(n^2)$ rounds after $a_1$ has entered the black hole.  \end{itemize}\qed\end{proof}

\begin{lemma}
The $LEADER$ following the algorithm \textsc{SingleEdgeBHSLeader()} correctly locates the black hole position.
\end{lemma}
\begin{proof}
We discuss all possible scenarios that can occur while executing algorithm \textsc{SingleEdgeBHSLeader()}.
\begin{description}
    \item [$a_1$ enters black hole at $u$ during \textit{cautious} walk:] In this case, according to Algorithm \ref{algorithmLeader}, if the edge $e=(u,v)$ between $LEADER$ and $a_1$ is missing, then $LEADER$ waits at $v$ until $e$ reappears. Meanwhile, $a_2$ is performing \textit{pendulum} walk, and since there is at most one dynamic edge at any round, hence $a_2$ will eventually reach $u$, i.e., enter the black hole from an alternate path. Now in this case, the $LEADER$ after waiting $2({length}_{a_2}+1+L_{a_2})$ rounds, concludes that $a_2$ has entered black hole. The conclusion is indeed correct, as one end of the edge which is missing edge, is occupied by $LEADER$ and the other end contains black hole. Since, there is at most one missing edge in $\mathcal{G}$ at any round. So, $a_2$ faces no obstruction while exploring the graph, until it enters the black hole. Hence, within $2({length}_{a_2}+1+L_{a_2})$ rounds $a_2$ otherwise will have reported back to $LEADER$. 
    \item [$a_1$ enters black hole during \textit{pendulum} walk:] $a_1$ which is initially performing \textit{cautious}, performs \textit{pendulum} walk for two reason.
    \begin{itemize}
        \item \underline{$a_2$ fails to report}. This situation arises when $a_1$ is initially performing \textit{cautious} walk with $LEADER$, and $a_2$ which is performing \textit{pendulum} walk fails to report. In this scenario, the $LEADER$ according to Algorithm \ref{algorithmLeader}, instructs $a_1$ to perform \textit{pendulum} walk, while it moves towards $a_2$ following $P_{a_2}$ and $PL_{a_2}$. Now we have two possibilities: first, $a_2$ is found and second $a_2$ is not found. If $a_2$ is found, then it is instructed to either perform \textit{cautious} or \textit{pendulum} walk, based on the fact that there is a forward missing edge or not. Now, in this situation, since $a_1$ has entered black hole and fails to report, then $LEADER$ moves towards $a_1$ while instructing $a_2$ to perform \textit{pendulum} walk. If the $LEADER$ does not find a missing edge corresponding to the last visited port of $a_1$, then it concludes that $a_1$ has entered black hole. Otherwise, if $LEADER$ gets stuck at a node $v$ due to a missing edge $e=(u,v)$, then it remains stationary, and while $a_2$ will either get stuck at $u$ or eventually enter black hole. So, ultimately $a_2$ also fails to report and $LEADER$ not knowing the reason behind $a_2$'s failure to report, moves towards $a_2$ leaving the node $v$. If $a_2$ is found, i.e., it got stuck at $u$ then it correctly concludes $a_1$ has entered the black hole, whereas if $LEADER$ encounters another missing edge while moving towards $a_2$, then after waiting an additional $2(L_{a_1}+{length}_{a_2}+1)$ round, it correctly concludes $a_1$ has entered black hole, as otherwise $a_1$ will have reported back to $LEADER$. Otherwise, if $a_2$ is not found, and there is no missing edge, then also $LEADER$ correctly concludes $a_2$ has entered black hole, as otherwise $a_2$ will have reported back to $LEADER$.  
        \item \underline{$a_1$ performs \textit{pendulum} walk with $a_2$}: This situation arised because, $a_2$ initially stopped reporting due to a missing edge, while $a_1$ is performing \textit{cautious} walk. Now, as $LEADER$ moves towards $a_2$ it instructs $a_1$ to change its movement to \textit{pendulum} walk. Moreover, it finds $a_2$ and the missing edge persists, in this situation, $a_2$ is further instructed to continue \textit{pendulum} walk whereas $LEADER$ remains stationary at one end of the missing edge. If both $a_1$ and $a_2$ fail to return, then $LEADER$ identifies at least one among them has entered black hole. It is because, the $LEADER$ holds one end of the missing edge, and at most one agent may be stuck at the other end, whereas the remaining agent must return if it has not entered the black hole. So, $LEADER$ moves towards the last reported agent $a_1$ (say w.l.o.g). If $a_1$ is not found and there is no missing edge in the forward direction, then $LEADER$ correctly identifies the black hole position. Otherwise, if $LEADER$ finds a missing edge along its path towards $a_1$ and further waits for $2({length}_{a_2}+1+L_{a_2})$ rounds, within which if $a_2$ also fails to return, then $LEADER$ concludes $a_2$ has entered black hole. Since, $LEADER$ has moved towards $a_1$ and encountered a new missing edge, this implies that the earlier missing edge has reappeared. Hence, $a_2$ has no other obstruction along its path towards $LEADER$, if it is originally stuck due to the earlier missing edge.
    \end{itemize}  
\end{description}
The explanation are similar, when $a_2$ enters black hole while performing either \textit{cautious} or \textit{pendulum} walk. So, we have shown that in each case the $LEADER$ correctly determines the black hole location. \qed\end{proof}

\begin{theorem}\label{complexity1edge}
   The agent following algorithms \textsc{SingleEdgeBHSAgent()} and \textsc{SingleEdgeBHSLeader()}, correctly locates the black hole in a dynamic cactus graph $\mathcal{G}$ with at most one dynamic edge at any round with $O(n^2)$ moves and in $O(n^2)$ rounds.
\end{theorem}
\begin{proof}
    Lemmas \ref{agentBH} and \ref{round} states that our Algorithms \ref{algorithmagent} and \ref{algorithmLeader}, correctly find the black hole in $O(n^2)$ rounds,  and since at each round each agent can move at most once. Hence, there can be at most 3 moves in each round of our algorithm. This implies, that the agent following Algorithms \ref{algorithmagent} and \ref{algorithmLeader}, solves the black hole search problem in $O(n^2)$ moves. 
\qed\end{proof}

\subsection{Black Hole Search in Presence of Multiple Dynamic Edges}\label{multiEdgeBHS-section}
In this section, we present an algorithm \textsc{MultiEdgeBHS()} for the agents to locate the black hole position, where the underlying graph is a dynamic cactus graph $\mathcal{G}$ but unlike the earlier section, where at most one edge can be missing at any round, in this section, we discuss the case in which there can be at most $k$ dynamic edges, such that the underlying graph remains connected. As discussed earlier, each node $v\in\mathcal{G}$ is equipped with a whiteboard of $O(deg(u)(\log deg(u)+k\log k))$ bits of memory. Moreover, there are a team of $2k+3$ agents, $\mathcal{A}=\{a_1,\cdots,a_{2k+3}\}$ which executes the algorithm \textsc{MultiEdgeBHS()}, starting from a safe node also termed as $home$. Next, we define the contents of information that can be present on a whiteboard. 

\noindent\textbf{Whiteboard:} For each node $v\in \mathcal{G}$, a whiteboard is maintained with a list of information for each port of $v$. For each port $j$, where $j\in \{0,\cdots,deg(v)-1\}$, an ordered tuple $(g_{1}(j),g_{2}(j))$ is stored on the whiteboard. The function $g_1$ is defined to be exactly the same as the function $f$ in section \ref{descsingleedge}.

On the other hand, the function $g_2$ is defined as follows, $g_{2}:\{0,\ldots,deg(v)-1\}\rightarrow \{\bot,0,1\}$,
\[g_{2}(j) = 
     \begin{cases}
       \text{$\bot$,} &\quad\text{if an agent is yet to visit the port $j$}\\
       \text{$0$,} &\quad\text{if no agent has returned to the node $v$ along $j$}\\
       \text{$1$,} &\quad\text{otherwise}\\ 
     \end{cases}\]

Each agent performs a $t$-\textsc{Increasing-DFS} \cite{fraigniaud2005graph}, where the movement of each agent can be divided into two types \textit{explore} and \textit{trace}:
\begin{itemize}
    \item In \textit{explore}, an agent performs either \textit{cautious} walk or \textit{pebble} walk depending on the situation. 
    \item In \textit{trace}, an agent walks along the safe ports of a node $v$, i.e., all such port $j\in v$ with $g_{2}(j)=1$. 
\end{itemize}
Our algorithm \textsc{MultiEdgeBHS()} requires no $LEADER$, unlike our previous algorithms in section \ref{descsingleedge}. In this case each $a_i$ (where $i\in\{a_1,a_2,\cdots,a_{2k+3}\}$) executes their operations, based on the whiteboard information they gather at each node. 

Next, we give a detailed description of the algorithm \textsc{MultiEdgeBHS()}. 

\noindent\textbf{Outline of Algorithm:} The team of agents $\mathcal{A}=\{a_1,\ldots,a_{2k+3}\}$ are initially located at a safe node, termed as $home$. Initially, the whiteboard entry corresponding to each port at each node in $\mathcal{G}$ is $(\bot,\bot)$. The lowest $Id$ agent present at $home$, i.e., $a_1$ in this case decides to perform \textit{cautious} walk. At the first round, it chooses the $0$-th port and if the edge corresponding to $0$-th port exists, then it moves along it to an adjacent node $v$ while updating $(g_1(0),g_2(0))$ at $home$ from $(\bot,\bot)$ to $(0\circ a_1,0)$. Now, if $v$ is safe, and the edge $(home,v)$ exists then it returns to $home$ at the second round, while updating $g_2(0)$ at $home$ to 1, i.e., marking the edge $(home,v)$ as safe. Further, if at the third round the edge $(home,v)$ exists, then $a_1$ accompanies $\mathcal{A}\backslash\{a_1\}$ to $v$ while updating $g_1(0)=g_1(0)\circ A'$, where $A'=\{a_2,\cdots,a_{2k+3}\}$. Otherwise, if the edge has gone missing, at the second round, then $a_1$ and $a_2$ remain at $v$ and $home$, respectively until the edge reappears, whereas the remaining agent continues to perform their respective movement. 

Now, consider a scenario where the edge $(home,v)$ goes missing at the third round when all the $\mathcal{A}$ agents are at $home$. In this case, $a_1$ remains at $home$ until the edge reappears, whereas the remaining agents continue to perform \textit{cautious} walk along the other available ports. Whenever the edge reappears suppose at the $r$-th round, then it starts \textit{pebble} walk. The movement of $a_1$ performing \textit{pebble} walk is as follows: at the $r+1$-th round $a_1$ moves to $v$, further the agent moves as follows:
\begin{itemize}
    \item If there exists a port $i$ with $g_2(i)=0$, and the edge remains, then at the $r+2$-th round $a_1$ stays at $v$. If no agent returns along $i$-th port, $a_1$ concludes that the node w.r.to the port $i$ is the black hole node. Otherwise, if an agent $a_j$ returns (for some $j>0$), then both $a_1$ and $a_j$ start \textit{cautious} walk. Moreover, if the edge does not exist, and there is no other agent at $v$, then $a_1$ waits until the edge reappears. Otherwise, if there is already an agent waiting, then $a_1$ decides to move from $v$ to some adjacent node, based on the whiteboard entry.
    \item If there exists a port $i$ with $g_2(i)=\bot$, then at $r+2$ round $a_1$ chooses that port and moves to the adjacent node while updating $(g_1(i),g_2(i))$ to $(0\circ a_1,0)$.
    \item If each port at $v$ is having its $g_2()$ value 1, then at $r+2$-th round $a_1$ backtracks to $home$, if $(home,v)$ exists, otherwise stays at $v$ until the edge reappears.
\end{itemize}

In general, whenever multiple agents meet at a node, they start \textit{cautious} movement. Moreover, when a single agent waiting for a missing edge reappears, then it starts \textit{pebble} walk. In addition, whenever an agent finds a port $i$ at some node in $\mathcal{G}$ with $g_2(i)=0$ and the edge exists, then after waiting for a round, it concludes the adjacent node w.r.to $i$ is the black hole node.

The pseudo code of \textsc{MultiEdgeBHS()} is explained in Algorithm \ref{algorithmLeader}.

\begin{algorithm2e}[!ht]\footnotesize
\caption{\sc{MultiEdgeBHS($u,a_i$)}}\label{multiedge-psuedo}
\If{more than one agent is available at $u$}
{
\textsc{MultipleAgent($u,a_i$)}.\\
}
\Else
{
Perform \textsc{SingleAgent($u,a_i$)}.\\

}
\end{algorithm2e}

\begin{algorithm2e}[!ht]\footnotesize
\caption{\textsc{MultipleAgent($u,a_i$)}}\label{multipleAgent-k}
\If{$a_i$ is the minimum $Id$ agent at $u$}
{
\If{a port $j$ with $(\bot,\bot)$ exists}
{
Traverse to $v$ via $j$ and update 
$g_1(j)=0\circ a_i$ and $g_2(j)=0$ at $u$.\\
\If{$v$ is safe and $(u,v)$ exits} 
{
Return to $u$, and update $g_2(j)=1$ at $u$.\\
\If{$(u,v)$ exists}
{
Move all available agents to $v$ via and update $g_1(j)=0\circ A'$, where $A'$ is the set of available agents at $u$.\\
Perform \textsc{MultiEdgeBHS($v,a_i$)}.\\
}
\Else
{
Stay at $u$ until $(u,v)$ reappears, if no agent is waiting for $(u,v)$. \\
Perform \textsc{MultiEdgeBHS($u,a_m$)}, where $a_m$ is the next available minimum $Id$ agent at $u$ after $a_i$.\\
}
}
\ElseIf{$v$ is safe but $(u,v)$ goes missing}
{
Remain at $v$, until the edge reappears.\\
}
}
\ElseIf{a port $j$ of $u$ with $(0\circ A,0)$ exists}
{
\If{the edge is missing}
{
Stay at $u$ until $(u,v)$ reappears, if no agent is waiting for $(u,v)$. \\
Perform \textsc{MultiEdgeBHS($u,a_m$)}, where $a_m$ is the next available minimum $Id$ agent at $u$ after $a_i$.\\
}
\Else
{Conclude the adjacent node to be the black hole and terminate.\\}
}
\ElseIf{a port $j$ of $u$ with $(0\circ A,1)$ exists}
{
\If{the edge is missing}
{
Stay at $u$ until $(u,v)$ reappears, if no agent is waiting for $(u,v)$. \\
Perform \textsc{MultiEdgeBHS($u,a_m$)}, where $a_m$ is the next available minimum $Id$ agent at $u$ after $a_i$.\\
}
\Else
{
\If{$A$ contains $Id$ of at least one available agent at $u$}
{Choose a different available port or backtrack.\\}
\Else
{
Move all available agents to $v$ via $j$-th port and update $g_1(j)=g_1(j)\circ A'$, where $A'$ are the set of all available agents then perform \textsc{MultiEdgeBHS($v,a_i$)}.\\
}
}
}
\Else
{
Backtrack with all available agents to a node $w$ with a port $j$, $g_1(j)\neq 1$ perform \textsc{MultiEdgeBHS($w,a_i$)}.\\
}
}
\Else
{
Follow instruction of the minimum $Id$ agent.\\
}

\end{algorithm2e}

\begin{algorithm2e}[!ht]\footnotesize
\caption{\textsc{SingleAgent($u,a_i$)}}\label{singleAgent-k}
\If{a port $j$ at $u$ with $(\bot,\bot)$ exists}
{
If that edge exists, then visit that node $v$ while updating $g_1(j)=0\circ a_i$ and $g_2(j)=0$.\\
Else, if that edge is missing, stay at $u$ until $(u,v)$ reappears, if no agent is waiting for $(u,v)$.\\
\If{$v$ is safe}
{
If the edge $(u,v)$ exists, then return to $u$ and update $g_2(j)=1$.\\
Else, if the edge $(u,v)$ is missing, then wait until it reappears.\\
\If{$(u,v)$ exists}
{
Visit $v$ and perform \textsc{MultiEdgeBHS($v,a_i$)}.\\
}
\Else
{
Stay at $u$ until $(u,v)$ reappears, if no agent is waiting for $(u,v)$. Else, choose a different port and continue \textsc{MultiEdgeBHS($u,a_i$)}.\\
}
}
}
\ElseIf{a port $j$ at $u$ with $(0\circ A,0)$ exists}
{
\If{$A$ does not contain $a_i$}
{
If the edge exists, and no agent returns after one round, then terminate by concluding the node w.r.to the port $j$ is the black hole.\\

Else, if the edge is missing, then stay at $u$ until $(u,v)$ reappears, if no agent is waiting for $(u,v)$. Else, choose a different port and perform \textsc{MultiEdgeBHS($u,a_i$)}. \\
}
\Else
{
Choose a different port and perform \textsc{MultiEdgeBHS($u,a_i$)}.\\
}
}

\ElseIf{a port $j$ at $u$ with $(0\circ A,1)$ exists} 
{
\If{$A$ does not contain $a_i$}
{Move along $j$ to $v$ while updating $g_2(j)=g_2(j)\circ a_i$ at $u$ and perform \textsc{MultiEdgeBHS($v,a_i$)}.\\}
\Else
{
Choose a different port and perform \textsc{MultiEdgeBHS($u,a_i$)}.\\
}
}
\Else {Backtrack to a node $w$ with a port $j$, $g_1(j)\neq 1$ and perform \textsc{MultiEdge($w,a_i$)}.\\}

\end{algorithm2e}

\subsubsection{Correctness and Complexity:} In this section we analyze the correctness and complexity of our algorithm \ref{multiedge-psuedo}.

\begin{lemma}\label{atmost2}
    Given a dynamic cactus graph $\mathcal{G}$ with at most $k$ dynamic edges at any round $r$. Our algorithm, \textsc{MultiEdgeBHS()} ensures that at most 2 agents are stuck due to a missing edge at any round.
\end{lemma}
\begin{proof}
    An agent executing algorithm \ref{multiedge-psuedo} may encounter a missing edge in two possible ways:
    \begin{itemize}
        \item First, when the agent is performing \textit{cautious} walk.
        \item Second, when the agent is performing \textit{pebble} walk.
    \end{itemize}
    In either case, the agent remains stationary when it encounters a vacant missing edge (refer to Algorithms \ref{multipleAgent-k} and \ref{singleAgent-k}). Moreover, if any agent at a node $u$ finds another agent waiting due to a missing edge w.r.to a port $j$ (where $j\in\{0,1,\cdots,deg(u)-1\}$), then it either chooses a different available port or if there are no such available ports then the agent backtracks. Hence, either end of a missing edge can be occupied by at most 2 agents.\qed\end{proof}

\begin{lemma}\label{bh-2agent}
    In the worst case at most 2 agents are consumed by the black hole, while the agents are following the algorithm \textsc{MultiEdgeBHS()}.
\end{lemma}

\begin{proof}
By lemma \ref{uv-path}, it has been shown that if $v$ is the black hole node and the node $u$ is the $home$, then any path from $home$ to $v$ must either pass along $v_0$ or $v_1$, where $v_0$ and $v_1$ belong to the same cycle as $v$ in $\mathcal{G}$. This implies that an agent passing through these unexplored edges $(v_0,v)$ and $(v_1,v)$ cannot mark the nodes $v_0$ and $v_1$ to be safe, as they enter the black hole. So, any subsequent agent visiting either $v_0$ or $v_1$ finds that $g_2()$ value corresponding to the edges $(v_0,v)$ and $(v_1,v)$ are 0. As any cycle can have at most one missing edge at any round such that the underlying graph remains connected, the adversary can either disappear $(v_0,v)$ or $(v_1,v)$ at any round. So, any set of agents trying to visit these unsafe nodes while executing \textsc{MultiEdgeBHS()}, can find these nodes $v_0$ and $v_1$ unsafe, and can successfully locate the black hole node without any further agents entering the black hole. Hence, this guarantees that at most two agents can be consumed by the black hole.
\qed    
\end{proof}

\begin{lemma}\label{infinitecycle-k}
    Our algorithm \textsc{MultiEdgeBHS()} ensures that no agent enters an infinite cycle.
\end{lemma}
\begin{proof}
    Consider $C$ to be a cycle in $\mathcal{G}$ and $u$ is the node in $C$ along which an agent $a_i$ while executing Algorithm \ref{multiedge-psuedo} traverses to an adjacent node $v\in C$ via the port $i$. Now, $a_i$ may be traversing with other agents performing \textit{cautious} walk (in that case $g_1(i)=g_1(i)\circ A'$ refer to step 28 of Algorithm \ref{multipleAgent-k}) or traversing alone performing \textit{pebble} walk (in this case $g_1(i)=g_1(i)\circ a_i$ refer to step 19 of Algorithm \ref{singleAgent-k}). In either case, $a_i$ after exploring $C$ whenever reaches $u$ and again decides to travel along $i$, then with the help of $g_1(i)$ at the whiteboard, identifies that it has already traversed this port before. In this situation, it decides not to take this port, and either chooses another available port or backtrack (refer to step 25-26 of Algorithm \ref{multipleAgent-k} when $a_i$ is performing \textit{cautious} walk, otherwise step 12-16 of Algorithm \ref{singleAgent-k}).\qed\end{proof}

    \begin{lemma}\label{explore-k}
        \textsc{MultiEdgeBHS()} ensures that any agent which is not stuck due to a missing edge can explore the remaining graph until it either enters the black hole or detects it.
    \end{lemma}
\begin{proof}
     As stated, each agent is a $t$-state finite automata, where $t \ge \alpha n\log\Delta$. An agent while executing either during \textit{cautious} or \textit{pebble} walk, performs the $t-$\textsc{Increasing-DFS} algorithm, which in turn ensures that any static graph of diameter at most $n$ and maximum degree $\Delta$ can be explored with $O(n\log\Delta)$ bits of memory (refer to Theorem 6 and to Corollary 7 in \cite{fraigniaud2005graph}). Now, in this situation of dynamic graph, the whiteboard in addition helps the agent in determining the ports it has yet to visit from a node ($g_1(i)=\bot$ or $0\circ A$, where $A$ does not contain $a_i$, refers that $i$-th port not visited by $a_i$). Moreover, the whiteboard also restricts the agent from entering in an infinite cycle loop (refer to Lemma \ref{infinitecycle-k}). So, an agent which is not stuck due to a missing edge, can explore each node of the underlying graph until it either enters the black hole or detects the black hole.\qed\end{proof}


\begin{theorem}\label{agentkedge}
    Given a dynamic cactus graph $\mathcal{G}$ with at most $k$ dynamic edges at any round. Our algorithm \textsc{MultiEdgeBHS()} ensures that it requires at most $2k+3$ agents, to successfully locates the black hole position. 
\end{theorem}
\begin{proof}
    By lemma \ref{atmost2}, we have shown that at most 2 agents are stuck due to a dynamic edge. Now at any round, there can be at most $k$ dynamic edges, which in turn implies that at most $2k$ agents can be stuck due to these dynamic edges. Moreover by lemma \ref{bh-2agent}, it is shown that at most 2 agents get eliminated by the black hole. Moreover, lemma \ref{explore-k}, ensures that any agent which is not stuck can explore the graph until it either detects the black hole or gets eliminated by it. Hence, the remaining agent among $2k+3$ agents can traverse along the graph unobstructively, and whenever it finds a vertex adjacent to a black hole, along which a port is marked unsafe (i.e. $g_2()$ value with respect to a port is $0$) in whiteboard, it waits for one round and if no agent returns then it correctly concludes that the node w.r.to the unsafe port is the black hole position.
\qed\end{proof}

\begin{theorem}\label{multi-complexity}
   The team of agents $\mathcal{A}=\{a_1,a_2,\cdots,a_{2k+3}\}$ following \textsc{MultiEdgeBHS()}, locates the black hole in a dynamic cactus graph $\mathcal{G}$ with at most $k$ dynamic edges at any round with $O(kn)$ rounds and in $O(k^2n)$ moves.
\end{theorem}

\begin{proof}
Consider a graph with $k$-cycles, $C_1,C_2,\cdots,C_k$, each cycle can have at most one dynamic edge at a round, such that the graph remains connected. Now, we analyze the complexity with the help of the following cases based on the size of $C_i$'s, $1\le i \le k$.
\begin{itemize}
    \item \textit{Case-1:} Suppose $|C_i|\le \frac{n}{k}, ~\forall ~1 \le i \le k$, consider a cycle $C_1$ with the nodes $u,v,v',w,w'\in C_1$. Let us assume, the set of $2k+3$ agents enter $C_1$ from $u$ and moves along a clockwise direction and encounters a missing edge $(v,v')$, which separates $a_1$ (say) from $\mathcal{A}$. In this scenario, the algorithm \textsc{MultiEdgeBHS()} instructs $a_2$ to remain stationary at $v$, whereas the remaining agents leave $a_1$ and $a_2$ and start moving in a counter-clockwise direction, either to \textit{trace} or to \textit{explore}. Now, in the meantime suppose the adversary revives the edge $(v,v')$, which reunites $a_1$ and $a_2$ and they continue to move forward, whereas disappears another edge $(w,w')$ in the counter-clockwise direction, which separates further two agents $a_3$ and $a_4$ (say) in a similar manner. Now, in order to separate these 4 agents from $\mathcal{A}$, the adversary requires $O(\frac{n}{k})$ rounds. Continuing in this manner, in order to separate $2k+3$ agents by reappearing and disappearing edges in $C_1$, it takes $O(k\frac{n}{k})=O(n)$ rounds. Hence, in order to explore each of $k$ such cycles, in the worst case $O(kn)$ rounds are required. Moreover, in each round at most $2k+3$ agents move, hence in the worst case $O(k^2n)$ moves are required.
    \item \textit{Case-2:} Suppose $|C_i|=3, ~\forall~ 2 \le i \le k$, and $|C_1|=n+2-3k$. As discussed in the earlier case it requires $O(k(n+2-3k))$ round to separate $2k+3$ agent in $C_1$, whereas since the other cycles are of size 3, it takes an additional $O(k)$ rounds to perform the same task in the remaining $k-1$ cycles. So, in general, it requires $O(k(n-2k))$ rounds to explore $\mathcal{G}$ in this case. Moreover, at most $2k+3$ agents move in each round, hence in the worst case $O(k^2(n-2k))$ moves. 
\end{itemize}
So, we analyze the two extreme possibilities, from the above cases, and conclude that in the worst case the algorithm  \textsc{MultiEdgeBHS()} requires $O(kn)$ rounds and $O(k^2n)$ moves to locate the black hole.\qed\end{proof}

\section{Conclusion}\label{conclusion}
In this paper, we studied the black hole search problem in a dynamic cactus for two types of dynamicity. We propose algorithms and lower bound and upper bound complexities in terms of number of agents, rounds and moves in each case of dynamicity. First, we studied at most one dynamic edge case, where we showed with 2 agents it is impossible to find the black hole, and designed a black hole search algorithm for 3 agents. Our algorithm is tight in terms of number of agents. Second, we studied the case when at most $k$ edges are dynamic. In this case, also we propose a black hole search algorithm with $2k+3$ agents. Further, we propose that it is impossible to find the black hole with $k+1$ agents in this scenario. A future work is to design an algorithm which has a tight bound in terms of number of agents when the underlying graph has at most $k$ dynamic edges. Further, it will be interesting to find an optimal algorithm in terms of complexity in both cases of dynamicity.

\bibliographystyle{splncs04}
\bibliography{bibliog.bib}
\end{document}